\documentclass[11pt]{article}

\usepackage[top=0.8in,bottom=0.8in,left=1in,right=1in]{geometry}
\usepackage{amsmath,amssymb,amsthm,mathtools}
\usepackage{booktabs}
\usepackage{enumitem}
\usepackage{microtype}
\usepackage[hidelinks]{hyperref}
\hypersetup{
  pdftitle={Two-State Max-Plus Comparison Is Decidable},
  pdfauthor={Keigo Oka}
}

\newtheorem{theorem}{Theorem}[section]
\newtheorem{lemma}[theorem]{Lemma}
\newtheorem{corollary}[theorem]{Corollary}
\newtheorem{remark}[theorem]{Remark}

\newcommand{\Zmax}{\mathbb{Z}_{\max}}
\newcommand{\minf}{-\infty}
\newcommand{\eval}[1]{[\![#1]\!]}

\title{Two-State Max-Plus Comparison Is Decidable}
\author{Keigo Oka\thanks{Google. Work done in a personal capacity; no Google internal resources were used. Email: \href{mailto:ogiekako@gmail.com}{ogiekako@gmail.com}. ORCID: \href{https://orcid.org/0009-0007-8119-9267}{0009-0007-8119-9267}.}}
\date{September 1, 2026}

\begin{document}
\maketitle

\begin{abstract}
Daviaud, Guillon, and Merlet proved that comparison of max-plus automata is undecidable under a fixed state bound of 553 and explicitly left the range from 2 to 552 states open. We resolve the two-state endpoint. More strongly, given an arbitrary finite max-plus automaton $A$ and a max-plus automaton $B$ with at most two states, it is decidable whether $\eval{A}(w)\leq \eval{B}(w)$ for every word $w$. The structural reason is a one-dimensional projective normal form for two-state dynamics. Outside an effective bounded region, a transition has one of three tail behaviors: it propagates the projective gap magnitude by a fixed affine shift with gap-independent height increment, forgets the gap with gap-independent height increment, or reads the gap magnitude into the height increment and then forgets it. In particular, any transition whose output depends on the unbounded gap necessarily destroys that gap. This yields an exact one-counter realization of $B$. Effective semilinearity of context-free Parikh images then reduces comparison to Presburger arithmetic. As a consequence, two-state max-plus comparison, equivalence, and positivity are decidable.
\end{abstract}

\section{Introduction}

Let
\[
\Zmax=(\mathbb Z\cup\{\minf\},\max,+)
\]
be the max-plus semiring. A finite max-plus automaton $A$ over an alphabet $\Sigma$ computes a series
\[
\eval{A}:\Sigma^*\longrightarrow \Zmax
\]
by taking the maximum weight of an accepting run.

The exact comparison problem asks, for two automata $A,B$, whether
\[
\eval{A}(w)\leq \eval{B}(w)\qquad\text{for all }w\in\Sigma^*.
\]
It is undecidable in general. Daviaud, Guillon, and Merlet \cite{DGM2017} proved undecidability already with a bounded number of states: their published bound is 553 states. They explicitly asked what happens between 2 and 552 states, noting that even the two-state case appeared difficult.

We prove the following.

\begin{theorem}[Two-state right-hand comparison]\label{thm:main}
Let $A$ be an arbitrary finite max-plus automaton over $\Zmax$, and let $B$ be a max-plus automaton with at most two states. It is decidable whether
\[
\eval{A}(w)\leq \eval{B}(w)
\qquad\text{for every }w\in\Sigma^*.
\]
\end{theorem}

This is stronger than merely resolving the two-state bounded-state endpoint, because only the right-hand automaton is required to have at most two states.

\begin{corollary}\label{cor:basic}
The following problems are decidable:
\begin{enumerate}[label=(\roman*)]
\item comparison of two max-plus automata with at most two states;
\item equivalence of two max-plus automata with at most two states;
\item $\mathrm{Pos}^2_k(\Zmax)$ for every finite alphabet size $k$.
\end{enumerate}
\end{corollary}

The proof is elementary apart from one standard formal-language theorem. A two-state forward vector has only one unbounded projective coordinate, namely the difference of its two entries. For sufficiently large absolute difference, each letter has a fixed tail form. The essential property is one-way: if the current height increment depends on the magnitude of the unbounded projective gap, then the successor projective state no longer depends on that magnitude. A step that propagates the gap has a gap-independent height increment. This no-read-while-propagating property permits exact simulation by one nonnegative counter. The accepting transition language of the synchronized simulation is context-free, so effective Parikh semilinearity \cite{Parikh1966} reduces the existence of a comparison violation to Presburger arithmetic.

We use slightly more general conventions than \cite{DGM2017}: arbitrary max-plus initial and final weights and words in $\Sigma^*$ are allowed here. The MFCS 2017 formulation uses Boolean-valued initial/final vectors and nonempty words. Restricting to that convention only simplifies the construction; excluding the empty word can also be enforced by one bit of finite control.

No complexity bound is claimed here.

\section{Two-state projective dynamics}

If $B$ has only one state, add a permanently dead second state; this leaves the computed series unchanged. Thus fix a two-state max-plus automaton $B$. Its forward row after reading a prefix is a vector
\[
x=(x_1,x_2)\in\Zmax^2.
\]
When $x\neq(\minf,\minf)$, define its height
\[
H(x)=\max(x_1,x_2).
\]
If both coordinates are finite, define the signed projective gap
\[
z(x)=x_1-x_2\in\mathbb Z.
\]
If exactly one coordinate is finite, use $+\infty$ or $-\infty$ for the corresponding infinite gap. After subtracting the common height, every finite projective state is therefore represented by
\[
(0,-n)\quad\text{or}\quad(-n,0),\qquad n\in\mathbb N.
\]
The unbounded projective information is one-dimensional.

Let a letter $\sigma$ act by
\[
M=\begin{pmatrix}a&b\\c&d\end{pmatrix}\in\Zmax^{2\times2}.
\]
For a finite signed gap, use the representative $x=(z,0)$. Direct max-plus multiplication gives
\begin{align}
(xM)_1&=\max(z+a,c),\\
(xM)_2&=\max(z+b,d).
\end{align}
Hence the successor signed gap is
\begin{equation}\label{eq:gap}
F_M(z)=\max(z+a,c)-\max(z+b,d),
\end{equation}
whenever both displayed maxima are finite, with the evident infinite-gap conventions otherwise.

Write
\[
A_M=\max(a,b),\qquad C_M=\max(c,d).
\]
Subtracting the input height gives the height cocycle
\begin{equation}\label{eq:cocycle}
\delta_M(z)=
\begin{cases}
\max(A_M,C_M-z), & z\geq0,\\
\max(A_M+z,C_M), & z\leq0,
\end{cases}
\end{equation}
whenever the successor is nondead.

\section{The tail trichotomy}

Let $E$ be the set of all finite transition, initial, and final weights of $B$. Choose
\begin{equation}\label{eq:K}
K=1+\max\{|\alpha-\beta|:\alpha,\beta\in E\},
\end{equation}
with $K=1$ if $E$ has fewer than two elements. Thus $K$ strictly dominates every finite transition breakpoint, every fixed tail shift below, and the final readout breakpoint.

Consider first the positive tail $z>K$. Every maximum in \eqref{eq:gap} then has a forced branch. The four active-row cases are summarized in Table~\ref{tab:tail}. The negative tail follows by exchanging the two current-state coordinates.

\begin{table}[ht]
\centering
\small
\begin{tabular}{@{}llll@{}}
\toprule
active row $(a,b)$ & successor projective behavior & height increment & class \\
\midrule
$a,b$ finite & fixed gap $a-b$ & constant $\max(a,b)$ & forget \\
exactly one finite & translation, or infinite gap & constant & propagate or forget \\
$a=b=\minf$ & fixed by $(c,d)$, independent of $z$ & $C_M-z$ & read-and-forget \\
\bottomrule
\end{tabular}
\caption{Tail behavior for $z>K$. A translation has gap magnitude $n'=n+s$ for a fixed integer $s$ (with the symmetric negative-tail form obtained by exchanging the current states).}
\label{tab:tail}
\end{table}

For completeness, if $a$ is finite then $z>|c-a|$ forces $z+a>c$; if $a=\minf$, the first successor coordinate is $c$. The same statement holds for $b,d$. If both $a,b$ are finite, both successor coordinates come from the active row, so the gap is the fixed value $a-b$ and $\delta_M(z)=\max(a,b)$. If exactly one of $a,b$ is finite, the active successor coordinate is $H$ plus a fixed weight. The other coordinate is either absent, giving an infinite gap, or comes from the inactive row and differs by $n+s$ for a fixed integer $s$; in either case the height increment is constant. Finally, if $a=b=\minf$, both successor coordinates come from the inactive row, so their projective class is independent of $z$ while \eqref{eq:cocycle} gives $\delta_M(z)=C_M-z$.

Thus there is a genuine third possibility that should not be conflated with reading: a step may forget the gap without using its magnitude in the height increment.

\begin{lemma}[Tail trichotomy]\label{lem:tail}
For every projective state with finite gap magnitude $n=|z|>K$, every nondead letter transition has exactly one of the following forms, with constants determined by the letter and the tail side:
\begin{enumerate}[label=(\roman*)]
\item \emph{propagate}: the successor has a finite gap whose magnitude depends affinely on the old magnitude, $n'=n+s$, for a fixed integer $s$ (and a fixed successor side), while the height increment is a fixed constant $c$;
\item \emph{forget}: the successor projective state is independent of $n$, while the height increment is a fixed constant $c$;
\item \emph{read-and-forget}: the successor projective state is independent of $n$, while the height increment is $c-n$.
\end{enumerate}
Here \emph{propagate} refers to preservation of the dependence on $n$: the shifted value $n'=n+s$ may be at most $K$ for a particular input, in which case the successor lies in the bounded region. In particular, whenever the height increment depends on the unbounded gap magnitude, that magnitude is not propagated to the successor.
\end{lemma}

\begin{proof}
The positive-tail case is the case split above, and the negative-tail case is obtained by exchanging the current-state coordinates. The choice of $K$ makes every relevant comparison strict and also gives $K>|s|$ for every propagate shift $s$.
Moreover, in cases (ii) and (iii), the successor projective state is either an infinite-gap state or a finite-gap state whose magnitude is strictly less than $K$. Indeed, every finite successor gap in those cases is the difference of two finite entries of the letter matrix, so its absolute value is $<K$ by the definition of $K$.
\end{proof}

\begin{remark}
For example, the all-zero matrix has a silent forget tail: for every sufficiently large gap, the successor gap is $0$ and the height increment is $0$. This is why the useful structural statement is not a literal retain/read dichotomy, but the prohibition on simultaneously reading an unbounded magnitude into the output and propagating that magnitude to the future.
\end{remark}

\begin{remark}
The statement is specific to the one-dimensional projective geometry of two states. With three states, two independent projective differences can coexist, and an update can read one combination while retaining another.
\end{remark}

\section{Exact compilation to one counter}

\begin{lemma}[Functional one-counter realization]\label{lem:counter}
From $B$ one can effectively construct a one-counter transducer $T_B$ such that, for every word $w$:
\begin{enumerate}[label=(\roman*)]
\item if $\eval{B}(w)=\minf$, then $T_B$ has no successful computation on $w$;
\item if $\eval{B}(w)\neq\minf$, then $T_B$ has a successful computation on $w$, and every successful computation on $w$ emits exactly $\eval{B}(w)$.
\end{enumerate}
The construction may be chosen deterministic. Each input letter triggers a finite-control routine built from unit counter increments, blocking decrements, zero tests, and fixed integer emissions. Such a routine has a unique finite execution from any reachable configuration, although its running length may depend on the counter value (for example when the counter is drained to zero).
\end{lemma}

\begin{proof}
Use finite-control states $B_z$ for every integer $-K\le z\le K$, infinite-gap states $I_+$ and $I_-$, and tail states $T_+$ and $T_-$. Their meanings are as follows. In $B_z$ the exact signed projective gap is $z$; in $I_+$ (respectively $I_-$) only the first (respectively second) coordinate is finite; and in $T_+$ or $T_-$ the gap is finite with the indicated sign and magnitude $n>K$.

The counter invariant is:
\[
\begin{array}{ll}
\text{in }B_z,I_+,I_- &: \text{the counter is }0,\\
\text{in }T_+,T_- &: \text{the counter is exactly }n=|z|>K.
\end{array}
\]
In addition, at every boundary between these finite-control routines, the total emission produced so far is exactly the height of the current forward vector.

If the initial forward vector is $(\minf,\minf)$, take $T_B$ to have no successful computation, and there is nothing to prove. Otherwise its height is defined; the initial forward vector contributes that height as a fixed emission and initializes the invariant.

If the current gap is bounded, the exact successor and height increment are obtained from a finite lookup table. If the lookup successor is a tail gap of magnitude $n'>K$, load this fixed value from zero by $n'$ unit increments; otherwise the counter remains zero. Because there are only finitely many bounded control states and letters, each such load is an effectively constructible finite routine. Its explicit size may depend on the numerical magnitudes of the weights; this is immaterial here because no complexity bound is claimed. If the current gap is infinite, the only surviving current coordinate determines the successor by another finite lookup. Suppose now that the counter stores a tail magnitude $n>K$. Apply Lemma~\ref{lem:tail}.

In a propagate case, emit the fixed constant $c$ and change the counter by the fixed shift $s$, using $|s|$ unit increments or blocking decrements. Since $n>K>|s|$, these decrements cannot block. Let $m=n+s>0$ be the resulting counter value.

To determine whether $m\le K$, use a deterministic destructive probe. For $j=0,1,\ldots,K$, after exactly $j$ probe decrements, test whether the counter is zero. If it is, then $m=j$; record that exact bounded gap in finite control and leave the counter at zero. If no zero test succeeds for $j\le K$, perform one further decrement and then exactly $K+1$ increments. This restores the original value $m$, including in the boundary case $m=K+1$, and the machine remains in the appropriate tail control state. All probe steps emit zero.

In a forget case, the old magnitude will never be needed again. Drain the counter to zero with zero emission, emit the fixed constant $c$, and move to the successor projective control state. By the final observation in the proof of Lemma~\ref{lem:tail}, that successor is either an infinite-gap state or a bounded finite-gap state, so counter value zero is the correct representation.

In a read-and-forget case, drain the counter to zero while emitting $-1$ per decrement, then emit the fixed constant $c$ and move to the successor projective control state. The total contribution of the routine is exactly $c-n$. Again, the successor is bounded or infinite by Lemma~\ref{lem:tail}, so the counter-zero invariant is restored. Dead transitions reject.

The final vector is handled by the same one-dimensional argument. Bounded and infinite projective states are finite-control lookups. Suppose, for example, that the current state is $T_+$, so the forward vector is $(H,H-n)$ with $n>K$, and let the final weights be $f_1,f_2$.

If $f_1$ is finite, then whenever $f_2$ is also finite,
\[
  n>K>|f_1-f_2|
\]
implies $f_1>f_2-n$. Hence the first coordinate necessarily realizes the final maximum, the final contribution is the fixed constant $f_1$, and the counter may be drained silently. If $f_1=\minf$ and $f_2$ is finite, the final contribution is $f_2-n$; drain the counter while emitting $-1$ per decrement and then emit $f_2$. If both final weights are $\minf$, reject. The case $T_-$ is symmetric.

All routines above are deterministic: at each zero test the zero and positive-counter cases are disjoint, and every other step is forced by the current finite control and input letter.

We prove correctness by induction on the processed input prefix. At each boundary between letter routines, the finite control and counter encode the exact projective class of the forward vector as specified above, and the total emission produced so far is exactly its height. The bounded, infinite, propagate, forget, and read-and-forget cases above preserve this invariant, while a dead max-plus successor has no successful machine continuation. The initialization establishes the invariant for the initial forward vector, and the final routine adds exactly the remaining final contribution. Consequently, if $\eval{B}(w)=\minf$ there is no successful computation, and otherwise the unique successful computation emits exactly $\eval{B}(w)$.
\end{proof}

\section{Deciding comparison}

\begin{proof}[Proof of Theorem~\ref{thm:main}]
First separate support. For a max-plus automaton, the set of words on which the value is finite is regular: ignore weights and retain exactly the finite-weight transitions, initial states, and final states. Hence it is decidable whether there exists a word with $\eval{A}(w)$ finite and $\eval{B}(w)=\minf$. Any such word is an immediate violation.

Henceforth discard every $\minf$-weighted transition of $A$, and in the explicit start/end encoding retain only finite initial and final weights. This does not change the series computed by $A$: every run using a $\minf$ weight has total weight $\minf$, while a word with no remaining accepting run still has value $\minf$. Thus every accepting run considered below has an ordinary integer weight.

Assume now that we work on the common finite support. Construct the functional one-counter transducer $T_B$ from Lemma~\ref{lem:counter} and view its finite-control routines at the level of their unit one-counter transitions; counter-dependent loops are represented by ordinary loops in this finite transition graph and are not unrolled to a fixed length. Synchronize each transition that consumes a letter $\sigma$ with a finite-weight transition of $A$ carrying the same letter, while internal counter transitions leave the $A$-state unchanged. Encode finite initial and final weights of $A$ by dedicated start and end transitions. The resulting finite one-counter machine nondeterministically chooses an accepting run $\rho$ of $A$, while the $T_B$ component has only the exact output $\eval{B}(w)$ on the synchronized word.

Give every transition $e$ of this finite product one-counter machine its own fresh symbol $\gamma_e$. Let $L$ be the language of fresh-symbol sequences labeling successful computations. Then $L$ is accepted by an effectively constructible ordinary one-counter automaton, hence by a pushdown automaton with one stack symbol above a bottom marker: increment pushes that symbol, blocking decrement pops it, and a zero test checks the bottom marker. Therefore an effective context-free grammar for $L$ can be constructed. By the constructive form of Parikh's theorem \cite{Parikh1966,KopczynskiTo2010}, its Parikh image
\[
  S=\Psi(L)\subseteq\mathbb N^m
\]
is effectively semilinear.

Enumerate the product transitions as $e_1,\ldots,e_m$. For each $e_i$, let $a_i\in\mathbb Z$ be its contribution to the chosen $A$-run weight ($0$ on internal transducer transitions), and let $b_i\in\mathbb Z$ be its transducer emission. The dedicated start/end transitions include all finite initial and final contributions. Hence for every successful trace $t\in L$, with $v=\Psi(t)$,
\[
  \operatorname{wt}_A(\rho_t)=\sum_{i=1}^m a_i v_i,
  \qquad
  \eval{B}(w_t)=\sum_{i=1}^m b_i v_i,
\]
where $w_t$ and $\rho_t$ are respectively the synchronized input word and the accepting run of $A$ represented by $t$. The second equality uses Lemma~\ref{lem:counter}.

We claim that, after the support check above, a comparison violation exists if and only if
\[
  \exists v\in S:
  \sum_{i=1}^m (a_i-b_i)v_i>0. \tag{*}
\]

For the forward implication, suppose $\eval{A}(w)>\eval{B}(w)$. The support check ensures that $\eval{B}(w)$ is finite. Since $w$ has only finitely many runs in the finite automaton $A$, some accepting run $\rho$ attains $\eval{A}(w)$. Synchronizing $\rho$ with the successful computation of $T_B$ on $w$ gives a trace $t\in L$. Its Parikh vector $v=\Psi(t)$ satisfies $(*)$.

Conversely, suppose $v\in S$ satisfies $(*)$. By the definition of $S=\Psi(L)$, there exists an actual successful trace $t\in L$ with $\Psi(t)=v$. This trace represents some word $w_t$ and accepting run $\rho_t$ of $A$, and $(*)$ gives
\[
  \operatorname{wt}_A(\rho_t)>\eval{B}(w_t).
\]
Therefore
\[
  \eval{A}(w_t)\ge \operatorname{wt}_A(\rho_t)>\eval{B}(w_t),
\]
so $w_t$ is a genuine comparison violation.

Finally, because $S$ is effectively semilinear, condition $(*)$ is an existential Presburger condition over an effectively given semilinear set and is decidable. Hence it is decidable whether any comparison violation exists.
\end{proof}

\begin{proof}[Proof of Corollary~\ref{cor:basic}]
Comparison follows immediately when both automata have at most two states. Equivalence is the conjunction of the two comparison directions. For positivity, let $A_0$ be the one-state automaton with constant value zero. Then
\[
\eval{A_0}\leq\eval{B}
\]
if and only if $\eval{B}(w)\geq0$ for every word $w$.
\end{proof}

\section{Position relative to prior work}

Daviaud, Guillon, and Merlet \cite{DGM2017} established the bounded-state undecidability frontier at 553 states and explicitly identified the interval $2,\ldots,552$ as open. Theorem~\ref{thm:main} closes its two-state endpoint. Daviaud and Johnson \cite{DJ2017} studied the same two-state max-plus model but addressed semigroup identities rather than positivity or comparison. Daviaud's survey \cite{Daviaud2020} reviews containment and equivalence for weighted automata and records no two-state solution. More recently, Daviaud, Purser, and Tcheng \cite{DPT2025} proved the Big-O (affine-domination) problem decidable and PSPACE-complete; their problem is explicitly a relaxation of exact containment, which remains undecidable in general.

The proof mechanism here is different from nearby decidable restrictions such as finitely ambiguous classes, unary-alphabet tropical automata, or determinisability problems: the alphabet and left-hand automaton are unrestricted, while the right-hand automaton has only two states. The decisive structure is the exact one-counter reduction furnished by Lemmas~\ref{lem:tail} and \ref{lem:counter}.

A targeted literature review conducted through September 1, 2026 did not identify an earlier resolution of the two-state comparison case or an equivalent arbitrary-left/two-state-right containment theorem. Bibliographic search cannot certify absolute historical absence, so this note makes no stronger priority claim than that dated search result.

\section{Reproducibility and provenance}

A maintained version of this manuscript, the public regression code, and the dated literature review are available in the Exact Interaction Geometry research repository:
\[
\text{\url{https://github.com/ogiekako/exact-interaction-geometry}}.
\]
The result arose within the EIG research programme: its residual/interface viewpoint led to isolating the one-dimensional projective gap and asking whether one transition can both expose an unbounded residual magnitude and preserve it for future computation. The theorem itself is ordinary weighted-automata mathematics and does not depend on EIG.

The companion regression script checks the projective formulas, exhaustively classifies a finite family of letter tails into propagate / forget / read-and-forget (including the all-zero silent-forget case), and compares direct max-plus evaluation with a separately coded projective/counter evaluator on hundreds of thousands of finite words. The infinite theorem rests on the proof above, not on finite enumeration.

\section*{AI-assisted research disclosure}

This work was developed through AI-assisted mathematical research. OpenAI reasoning models, including GPT-5.6 Sol, materially contributed to theorem discovery, proof drafting, verifier generation, and literature-review support. The author directed the research programme and takes responsibility for the manuscript. Separate model-assisted review was used during drafting, but that review is provenance rather than mathematical evidence. The result is intended to stand on the written proof above; the public regression code is a finite regression of its algebraic mechanism and does not independently establish the theorem or its historical priority.


\begin{thebibliography}{9}
\small
\setlength{\itemsep}{0.35em}

\bibitem{DGM2017}
Laure Daviaud, Pierre Guillon, and Glenn Merlet.
\newblock Comparison of Max-Plus Automata and Joint Spectral Radius of Tropical Matrices.
\newblock In \emph{MFCS 2017}, LIPIcs 83, Article 19, 2017.
\newblock DOI: \href{https://doi.org/10.4230/LIPIcs.MFCS.2017.19}{10.4230/LIPIcs.MFCS.2017.19}.

\bibitem{DJ2017}
Laure Daviaud and Marianne Johnson.
\newblock The Shortest Identities for Max-Plus Automata with Two States.
\newblock In \emph{MFCS 2017}, LIPIcs 83, Article 48, 2017.
\newblock DOI: \href{https://doi.org/10.4230/LIPIcs.MFCS.2017.48}{10.4230/LIPIcs.MFCS.2017.48}.

\bibitem{Daviaud2020}
Laure Daviaud.
\newblock Containment and Equivalence of Weighted Automata: Probabilistic and Max-Plus Cases.
\newblock In \emph{LATA 2020}, LNCS 12038, pp. 17--32, 2020.
\newblock DOI: \href{https://doi.org/10.1007/978-3-030-40608-0_2}{10.1007/978-3-030-40608-0\_2}.

\bibitem{DPT2025}
Laure Daviaud, David Purser, and Marie Tcheng.
\newblock The Big-O Problem for Max-Plus Automata is Decidable (PSPACE-Complete).
\newblock \emph{Logical Methods in Computer Science}, 21(3), 3:1--3:35, 2025.
\newblock DOI: \href{https://doi.org/10.46298/lmcs-21(3:3)2025}{10.46298/lmcs-21(3:3)2025}.

\bibitem{Parikh1966}
Rohit J. Parikh.
\newblock On Context-Free Languages.
\newblock \emph{Journal of the ACM}, 13(4):570--581, 1966.

\bibitem{KopczynskiTo2010}
Eryk Kopczy\'nski and Anthony Widjaja To.
\newblock Parikh Images of Grammars: Complexity and Applications.
\newblock In \emph{LICS 2010}, pp. 80--89, 2010.
\newblock DOI: \href{https://doi.org/10.1109/LICS.2010.21}{10.1109/LICS.2010.21}.

\end{thebibliography}
\end{document}